\documentclass{ipart_v1}
\usepackage{amsbsy,enumerate}

\Year{201?}
\Vol{}
\Issue{}

\firstpage{1}

\usepackage{amssymb}
\usepackage{amsfonts}
\usepackage{mathrsfs}
\usepackage{bbm}
\usepackage{amsthm}
\usepackage{fontenc}
\usepackage{latexsym}
\usepackage{amsmath}
\usepackage{amsxtra}
\usepackage{amstext}
\usepackage{txfonts}

\usepackage{verbatim}

\newtheorem{theorem}{Theorem}[section]

\newtheorem{lemma}[theorem]{Lemma}

\newtheorem{remark}[theorem]{Remark}

\numberwithin{equation}{section}

\newcommand{\thmref}[1]{Theorem~\ref{#1}}

\newcommand{\lemref}[1]{Lemma~\ref{#1}}

\newcommand{\remref}[1]{Remark~\ref{#1}}

\begin{document}
\title[Harmonic spinors on axisymmetric 3-manifolds $\cdots$]{
Harmonic spinors on axisymmetric 3-manifolds with Melvin ends}

\author[]{A.K.M. Masood-ul-Alam$^{\flat}$, Qizhi Wang$^{\sharp}$}

\begin{abstract}
We prove the existence of harmonic spinor
fields in axisymmetric Riemannian 3-manifolds having nonnegative
scalar curvature and asymptotic to the usual constant time
hypersurface of Melvin's magnetic universe. Such a spinor can be
used in the proof of the uniqueness of the magnetized
Schwarzschild solution.
\end{abstract}

\maketitle


\section{Introduction}
It appears that there are not many works on the Dirac equation on
an asymptotically cylindrical Riemannian manifold, although the
Laplacian and similar elliptic differential operators on complete
manifolds with warped cylindrical ends have been studied (Lockhart
and McOwen
\cite{LM}, Ma and McOwen
\cite{MaMc} and references cited therein). On the other hand
harmonic spinors are important tools for proving rigidity and
uniqueness results in differential geometry and in the study of
time-symmetric black-hole solutions of Einstein equation. Many of
the physical problems in potential theory originate under the
assumption of isolated bodies surrounded by empty space and under
this setting asymptotically flat assumption naturally arises.
Strictly speaking in the physical world existence of a magnetic
field is more natural than empty space. Speaking very roughly the
solenoidal nature of the magnetic field relates to lack of
asymptotic spherical symmetry. Also one of the reasons for
assuming asymptotical flatness is that it can provide finite
energy and hence stable solutions. Magnetic fields provide an
example of an infinite energy solution stable under radial
perturbations. Thus in the Melvin magnetic universe (MMU) solution
of Einstein-Maxwell equations diverging coaxial tubes of magnetic
lines of force are held together by gravitational attraction in
such a way that under a radial perturbation they neither collapse
nor explode but settle down to the original solution. MMU is not a
finite energy solution using a physically reasonable definition of
energy. Its \textquotedblleft constant
 time" hypersurface is not an asymptotically
flat 3-manifold. We want to solve the Dirac equation on a
Riemannian 3-manifold which is in some sense asymptotically
Melvin. An asymptotically flat Riemannian 3-manifold has a concept
of mass when the decay to flatness is reasonably rapid (Bartnik
\cite{Bart}). This concept is of differential geometric origin though
it has been discovered in the study of physics and it corresponds
to the energy of an appropriate spacetime. The positive mass
theorem of Schoen and Yau
\cite{SY} says that the mass is nonnegative if the scalar curvature
is nonnegative and gives a rigidity result for $\mathbb{R}^{3}$
when the mass is zero.

Thorne \cite{T1,T2} has defined a concept of energy (C-energy) for
a finite region of a cylindrical spacetime. Radinschi and Yang
\cite{RY} considered another concept of such (quasi-local)
energy for MMU. Although MMU has infinite energy comparing it with
the magnetized Schwarzschild solution an asymptotically Melvin
(defined rigorously later) Riemannian 3-manifold can be assigned a
mass-like parameter using the decay coefficient of the metric. It
is the coefficient of a term of faster decay. While in the absence
field equations giving more information on the Ricci curvature we
could not prove the positivity of this decay coefficient assuming
only the nonnegativity of the scalar curvature,
Weitzenb\"ok-Lichnerowicz identity applied to the harmonic spinor
of appropriate decay gives a positive mass type theorem involving
the decay coefficient of the term of slower decay. This parameter
gives the magnetic field in the physical spacetime. Our formula
gives this asymptotic parameter in terms of the harmonic spinor
weighted scalar curvature integral and the integral for the norm
of the covariant derivative of the spinor.

In this paper we have not considered higher dimensions and all
types of cylindrical ends. We also did not look for optimal
generalizations of the parameters of the weighted Sobolev spaces
involved. These generalizations and related boundary value
problems will be studied elsewhere by one of the authors (QW).
Here we want to show the steps involved in our proofs by
considering not too technical problems and we keep the article
readable by multi-disciplinary researchers. One application of the
existence of a harmonic spinor we prove here occurs in the proof
of the uniqueness of the magnetized Schwarzschild solution in
\cite{Mas} where this existence is assumed without proof. In fact
existence of such harmonic spinors may be useful in the hitherto
unsolved problems of extending the black hole uniqueness theorems
(\cite{Wel}) in many other magnetized worlds.

\section{Preliminaries}
We consider a smooth complete Riemannian $3$-manifold $\left(
\Sigma ,\hat{g}
\right) $ having the metric of the form
\[
\hat{g}=\bar{g}+Xd\phi ^{2}
\]
where the function $X$ and the 2-metric $\bar{g}$ are independent
of the coordinate $\phi.$ $(\partial /\partial
\phi)$ is a Killing vector field and $\phi$ is a globally defined function
except on the fixed point set of the isometry. We assume that the
orbits of this Killing vector field are closed with period $2\pi.$
We shall call such a metric axisymmetric and the fixed point set
is the axis of symmetry. In general we say a complete Riemannian
$3$-manifold with metric $h$ has ends if outside a compact subset
$K$ it is a finite union of disjoint sets $U_{i}$ each
diffeomorphic to $\mathbb{R}^{3}\setminus B,$ $B$ being a closed
ball. Here $i$ counts the number of ends. Let $b$ be a nonnegative
constant. We say that $\hat{g}$ is asymptotically Melvin with
parameter $b$ if at any end $U$ w.r.t. the Euclidean spherical
coordinates $\left\{ r,\theta ,\phi
\right\} $ in $\mathbb{R}^{3}\setminus B,$ $h$ has the form of $\hat{g}$
with
\begin{eqnarray}
\bar{g} &=&(1+v_{1})F^{2}(dr^{2}+r^{2}d\theta ^{2}), \label{1}\\ X &=&( 1+v_{2})
F^{\;-2}r^{2}\sin ^{2}\theta \label{2}
\end{eqnarray}
where
\begin{equation*}
  F=1+b r^{2}\sin ^{2}\theta,
\end{equation*}
and $v_{1},v_{2}\in W_{-\tau +1}^{2,p}$ for large enough $q$ and
$\tau >1/2.$ For the existence of harmonic spinors the
restrictions on $p$ and $\tau$ are not optimal. We assume $p\geq
4$ so that geodesic equation at each point of the tangent bundle
has a unique solution. We shall define the weighted Sobolev spaces
later. They say for $i=1,2$
\begin{equation*} \label{v12}
v_{i}=O( r^{-1}) ,\partial v_{i}=O( r^{-2}),
\text{ }\partial ^{2}v_{i}\in L_{-\tau -2}^{p}(
\mathbb{R}^{3}\setminus B)
\end{equation*}
where the last conditions imply that in some sense $\partial
^{2}v_{i}=o( r^{-\tau-2}).$ In some applications we can take out
the $O(r^{-1})$ term from $v_{1}$ and then we assume
\begin{equation} \label{v3}
\bar{g} =(1+v_{3})F^{2}\left(
( 1-2M/r) ^{-1}dr^{2}+r^{2}d\theta ^{2}\right)
\end{equation}
where $M$ is a constant and $v_{3}\in W_{-\tau }^{2,p}.$

$W^{2,p}_{\delta}\equiv W^{2,p}_{\hat{g},\delta}$ is the space of
measurable $\mathbb{C}^{2}$ valued functions $u$ in $L^{p}_{\rm
loc}$ such that
\begin{equation} \label{sbspace1}
\int_{\Sigma,\hat{g}}\left|\dfrac{\partial^{|l|}}{\partial x^{l_{1}}\partial
x^{l_{2}}\partial x^{l_{3}}}u\right|^{p}\left(\sqrt{1+r^{2}}\
\right)^{-\delta p+|l|p-3}
\end{equation}
are finite for $|l|=0,1,2.$ Here $l=\{l_{1},l_{2},l_{3}\}$ is the
multi-index. For $u=(u^{1},u^{2})\in
\mathbb{C}^{2}$ we have $|u|^{2}=u^{1}\overline{u^{1}}+u^{2}\overline{u^{2}}.$ $\dfrac{\partial^{|l|}}{\partial x^{l_{1}}\partial
x^{l_{2}}\partial x^{l_{3}}}$ means an $l$-th weak partial
derivative. Powers of $\sqrt{1+r^{2}}$ is included to specify the
decay rate as $r\rightarrow \infty.$ Since $r$ is not defined
outside $U$ we interpret this factor to be a positive function in
the subset $\Sigma\setminus U.$

By Melvin's 3-metric we mean the following metric.
\begin{equation*}
  g_{\text{MMU}}=\left(1+(1/4)B^{2}r^{2}\sin^{2}\theta\right)^{2}(dr^{2}+r^{2}d\theta^{2})+
  \left(1+(1/4)B^{2}r^{2}\sin^{2}\theta\right)^{-2}r^{2}\sin^{2}\theta
  d\phi^{2}.
\end{equation*}
It is the $3$-metric induced on the \textquotedblleft constant
time" hypersurface of Melvin's magnetic universe. Since we shall
work mainly with $SU(2)$ spinors (two-component spinors) in
$3$-dimension we give the definition of weighted Sobolev spaces
for $\mathbb{C}^{2}$-valued functions. We shall denote the spaces
for the real-valued functions by the same notation because it will
be clear from the context what we mean. Since we shall be working
with spinors for two Riemannian metrics $\hat{g}$ and $g$ it is
better to stress that we shall be using the same $\mathbb{C}^{2}$
fields for our spinors. Thus locally the \textquotedblleft square"
of a $\mathbb{C}^{2}$ field $\xi$ will correspond to two local
vector fields on the manifold both vector fields having the same
components but in two different sets of orthonormal basis vector
fields, orthonormal respectively in the metrics $g$ and $\hat{g}.$
We shall denote by $\{\hat{e}_{i}\}$ the orthonormal basis for
$\hat{g}$ and by $\{e_{i}\}$ the orthonormal basis for $g.$ We do
not need to explicitly state $\xi$ as $g$-spinor or
$\hat{g}$-spinor because the operators acting on it will have the
subscript $g$ or $\hat{g}.$ Thus when we write $D_{g}\xi$ we
understand that $\xi$ is a $g$-spinor. For spinor
$\xi=(\xi^{1},\xi^{2})\in \mathbb{C}^{2}$ the norm is denoted by
$||\xi||^{2}=\xi^{1}\overline{\xi^{1}}+\xi^{2}\overline{\xi^{2}}.$
The integral in Eq.~\eqref{sbspace1} has the volume form of
$\hat{g}.$ Thus we define the $W^{2,p}_{\delta}$ norm for
functions on $(\Sigma,\hat{g})$ by
\begin{equation}
||u||_{\hat{g},2,p,\delta}=\sum\limits_{|l|=0}^{|l|=3}\left(\int\left|\dfrac{\partial^{|l|}}
{\partial x^{l_{1}}\partial x^{l_{2}}\partial
x^{l_{3}}}u\right|^{p}\left(\sqrt{1+r^{2}}\ \right)^{-\delta
p+|l|p-3}\sqrt{\det{\hat{g}}}\right)^{\dfrac{1}{p}}.
\end{equation}
$||u||_{\hat{g},p,\delta}$ will denote the weighted Lebesgue norm
$L^{p}_{\hat{g},\delta}.$

We also need to consider Sobolev norm $||u||_{2,p,\delta}$
relative to the following asymptotically flat metrics
corresponding to Eqs.~(\ref{1},\ref{v3}):
\begin{eqnarray}
g&=&(1+v_{1})[( dr^{2}+r^{2}d\theta
^{2}]+(1+v_{2})r^{2}\sin^{2}\theta d^{2}\phi,
 \label{AF1}\\ g&=&(1+v_{3})[(
1-2M/r)^{-1}dr^{2}+r^{2}d\theta ^{2}]+(1+v_{2})r^{2}\sin^{2}\theta
d^{2}\phi. \label{AF2}
\end{eqnarray}
In the rest of this section and in the next section we shall use
Eq.~\eqref{AF1} for $g.$ Since the $3$-measure of $\hat{g}$ in $U$
is
\begin{eqnarray}\label{3measureghat}
  \sqrt{\det{\hat{g}}}d\theta d\phi dr&=&F\sqrt{\det{g}}d\theta d\phi dr, \qquad (1\leq F\leq
  C_{1}(1+r^{2})),\nonumber\\
  &=&F(1+v_{1})^{1/2}(1+v_{2})^{1/2}r^{2}\sin\theta d\theta d\phi
  dr,
\end{eqnarray}
functions in $W^{2,p}_{\delta}$ with finite
$||\cdot||_{\hat{g},2,p,\delta}$ norm will be in
$W^{2,p}_{\delta+2/p}$ with finite $||\cdot||_{2,p,\delta+2/p}$
norm for the asymptotically flat metric. These two norms are
however not equivalent because $\sin\theta=0$ kills the $r^{2}$
growth in $F.$ Now to show the existence of a harmonic spinor
suitable for proving the positive mass theorem for an
asymptotically flat Riemannian $3$-manifold Bartnik showed
(Proposition 6.1 in
\cite{Bart}, minding a typo) that the Dirac operator is an
isomorphism from $W^{2,p}_{-\delta}$ onto $W^{1,p}_{-\delta-1}$
where $\delta\in (0,2).$ So for an asymptotically Melvin manifold
with the metric independent of $\phi$ we seek harmonic spinors in
the Sobolev spaces
\begin{equation*}
\mathbb{W}^{2,p}_{-\delta}=\{\xi\in
W^{2,p}_{-\delta}\mid\delta\in(2/p,2-2/p),
\partial_{\phi}\xi=0\}.
\end{equation*}
We now state the main result. It is proved in the next section.
\begin{theorem} \label{existthm}
Suppose the scalar curvature $R_{\hat{g}}\geq 0$ and $\left(
\Sigma ,\hat{g}\right) $ is asymptotically Melvin with parameter $b>0$ having finite number of
asymptotic ends. Let $\epsilon\in (2/p,2-2/p).$ Then the Dirac
operator $D_{\hat{g}}:\mathbb{W}^{2,p}_{-\epsilon}\longrightarrow
 \mathbb{W}^{1,p}_{-\epsilon-1-2/p}$ is an isomorphism onto the
 range of $D_{\hat{g}}.$
\end{theorem}
If $b=0$ we cannot decrease the exponent in the
 target space by $-2/p$ in the proof presented below. As a result we get the isomorphism $D_{\hat{g}}:\mathbb{W}^{2,p}_{-\epsilon}\longrightarrow
 \mathbb{W}^{1,p}_{-\epsilon-1}.$ If $b=0$ then $\hat{g}$ is also asymptotically
 flat.

\section{Existence of a harmonic spinor}
We prove \thmref{existthm} using the method of Bartnik
\cite{Bart} and the fact that Fredholm property is the same for two
bounded linear operators sufficiently close.
\begin{proof}[Proof of \thmref{existthm}.]
First we show that the Dirac operator of $\hat{g},$ namely
$D_{\hat{g}},$ is a small perturbation in some appropriate sense
of a
\textquotedblleft weighted" Dirac operator of the asymptotically
flat metric $g$ given in Eq.~\eqref{AF1} so that arguments as in
Theorem 1.10 of Bartnik
\cite{Bart} prove the semi-Fredholm property of the $D_{\hat{g}}$ in case
 the later operator denoted $D_{g}$ is semi-Fredholm.
For simplicity let us first assume that we have only one end $U.$
Let $\{e^{i}\}$ be an orthonormal frame field of $1$-forms
relative to $g=F^{-2}\bar{g}+F^{2}Xd\phi ^{2}$ where we used
Eqs.~(\ref{1},\ref{AF1}). In $U$ we choose,
\begin{equation*}
  e^{1}=(1+\nu_{1})^{1/2}dr, e^{2}=(1+\nu_{1})^{1/2}rd\theta, e^{3}=(1+\nu_{2})^{1/2}r\sin\theta
  d\phi.
\end{equation*}
Then we see that $\{\hat{e}^{i}\}$ where $\hat{e}^{i}=Fe^{i}$ for
$i=1,2$ , and $\hat{e}^{3}=F^{-1}e^{3}$ will be an orthonormal
frame field of $1$-forms relative to $\hat{g}.$ We note that for
the vectors of the dual frames, $\hat{e}_{i}=F^{-1}e_{i}$ for
$i=1,2$ , and $\hat{e}_{3}=Fe_{3}.$ We note that
\begin{equation*}
  \hat{g}_{rr}=(1+v_{1})F^{2}=F^{2}g_{rr}, \quad
\hat{g}_{\theta\theta}=(1+v_{1})F^{2}r^{2}=F^{2}g_{\theta\theta}, \quad
\hat{g}_{\phi\phi}=F^{-2}g_{\phi\phi}.
\end{equation*}
Let $\Gamma,\hat{\Gamma}$ be the Christoffel symbols of the metric
$g,\hat{g}$ in coordinates $\{r,\theta,\phi\},$ where $A=1,2$
corresponding to $r,\theta.$
\begin{eqnarray*}
&&\Gamma_{\phi \phi }^{\phi }=0, \, \Gamma_{AB}^{\phi }=0,
\, \Gamma_{B\phi }^{A}=0,
\, \Gamma_{\phi A}^{\phi }=(1/2)(
\partial \ln (XF^{2})/\partial x^{A}),\\
&&\hat{\Gamma}_{\phi \phi }^{\phi }=0, \, \hat{\Gamma}_{AB}^{\phi
}=0,
\, \hat{\Gamma}_{B\phi }^{A}=0, \,
\hat{\Gamma}_{\phi A}^{\phi }=(1/2)(
\partial \ln X/\partial x^{A}).
\end{eqnarray*}
We now give the relation between the
 \textquotedblleft connection coefficients"
$C_{mij}=\left\langle e_{m},\nabla_{e_{i}}e_{j}\right\rangle_{g}$
and $\hat{C}_{mij}=\left\langle
\hat{e}_{m},\hat{\nabla}_{\hat{e}_{i}}\hat{e}_{j}\right\rangle_{\hat{g}}.$
On $C,\hat{C}$ the indices refer to the corresponding frame
fields. We take $A=1,2.$ Since $C_{ABC}$ and $\hat{C}_{ABC}$ are
antisymmetric in $A,C$ and $C_{3 AB}=0=C_{A3 B}=\hat{C}_{3
AB}=\hat{C}_{AB3 }$ we need to compute only
$C_{2B1},C_{33B},\hat{C}_{2B1}$ and $\hat{C}_{33B}.$ We have
\begin{eqnarray*}
  \hat{C}_{211} &=& F^{-1}C_{211}+
bO(F^{-2}\sin\theta\cos\theta) \text{ with } C_{211}=O(r^{-2}).\\
 \hat{C}_{221}  &=& F^{-1}C_{221}+ bO(F^{-2}r\sin^{2}\theta)
\text{ with } C_{221}=O(r^{-1}). \\
  \hat{C}_{332 } &=&F^{-1}C_{332}+bO(F^{-2}r\sin\theta\cos\theta).\\
  \hat{C}_{331 } &=&F^{-1}C_{331}+bO(F^{-2}r\sin^{2}\theta).
\end{eqnarray*}
 Recalling that on the elements of $\mathbb{C}^{2},$
Clifford multiplications $\hat{e}^{i}\hat{\cdot}$ and $e^{i}\cdot$
both means multiplication by the same ($\sqrt{-1}$ times) Pauli
matrix $\sqrt{-1}\sigma_{k},$ we see
\begin{eqnarray*}
D_{\hat{g}}&=&\hat{e}^{i}\hat{\cdot}\hat{\nabla}_{\hat{e}_{i}}
=\hat{e}^{i}\hat{\cdot}\left(\partial_{\hat{e}_{i}}-(1/4)\left\langle
\hat{e}_{m},\hat{\nabla}_{\hat{e}_{i}}\hat{e}_{j}\right\rangle_{\hat{g}} \hat{e}^{m}\hat{\cdot}
\hat{e}^{j}\hat{\cdot}\right)
\nonumber \\&=&
e^{i}\cdot\left(\partial_{\hat{e}_{i}}-(1/4)\left\langle
\hat{e}_{m},\hat{\nabla}_{\hat{e}_{i}}\hat{e}_{j}\right\rangle_{\hat{g}} e^{m}\cdot
e^{j}\cdot\right)\nonumber
\\&=&
e^{A}\cdot\left(\partial_{\hat{e}_{A}}-(1/4)\hat{C}_{mAj}
e^{m}\cdot
e^{j}\cdot\right)+e^{3}\cdot\left(\partial_{\hat{e}_{3}}-(1/4)\hat{C}_{m3j}
e^{m}\cdot e^{j}\cdot\right)\nonumber
\\&=&
e^{A}\cdot\left(\partial_{\hat{e}_{A}}-(1/4)\hat{C}_{BAC}
e^{B}\cdot
e^{C}\cdot\right)+e^{3}\cdot\left(\partial_{\hat{e}_{3}}\right)+O(r^{-1})
\text{ since }\hat{C}_{33A}=O(r^{-1}) \nonumber
 \\&=&\sum_{i=1,2}F^{-1}{e}^{i}\cdot\nabla_{e_{i}}
+F{e}^{3}\cdot\partial_{e_{3}}+
O(r^{-1})+bO(F^{-2}r\sin^{2}\theta).
\end{eqnarray*}
Since $\int\limits_{0}^{\pi}F^{-2}r\sin^{2}\theta d\theta=(\pi/2)
r(1+br^{2})^{-3/2},$ in the average sense
$O(F^{-2}r\sin^{2}\theta)$ is $o(r^{-1}).$ Using the notation
$O(r^{-1})+bO(F^{-2}r\sin^{2}\theta)=O^{\prime}(r^{-1})$ we thus
have
\begin{equation}\label{oldeqn9}
D_{\hat{g}}=\sum_{i=1,2}F^{-1}{e}^{i}\cdot\nabla_{e_{i}}
+F{e}^{3}\cdot\partial_{e_{3}}+ O^{\prime}(r^{-1}).
\end{equation}
We denote the operator
$\sum_{i=1,2}F^{-1}{e}^{i}\cdot\nabla_{e_{i}}+F{e}^{3}\cdot\partial_{e_{3}}$
by $P$. We note that $P=F^{-1}D_{g}+O(r^{-1}):
\mathbb{W}^{2,p}_{-\epsilon}\longrightarrow
\mathbb{W}^{1,p}_{-\epsilon-1-2/p}$ since spinors in $\mathbb{W}$ spaces are independent of $\phi,$
and $F^{-1}C_{33A}=O(r^{-1}).$
 Since $g$ is asymptotically flat, $D_{g}:\mathbb{W}^{2,p}_{-\epsilon}\longrightarrow
\mathbb{W}^{1,p}_{-\epsilon-1}$ is a
semi-Fredholm operator by Proposition 6.1 in \cite{Bart}. As
$F^{-1}$ is bounded, so $P$ is a bounded linear operator. To check
the semi-Fredholm property of $P,$ we first show that
$F^{-1}D_{g}$ is semi-Fredholm. First we note that $\ker
F^{-1}D_{g}=\{\xi
\in \mathbb{W}^{2,p}_{-\epsilon}: F^{-1}D_{g}\xi=0\}$ is finite
dimensional because $\ker D_{g}$ is finite dimensional. Second we
note that range of $F^{-1}D_{g},$
$F^{-1}D_{g}\left(\mathbb{W}^{2,p}_{-\epsilon}\right)$ is closed
in $\mathbb{W}^{1,p}_{-\epsilon-1}.$ This is because the range of
$F^{-1}D_{g}$ is the set of spinors in the range of $D_{g}$
multiplied by the bounded function $F^{-1},$ and the range of
$D_{g}$ is closed in $\mathbb{W}^{1,p}_{-\epsilon-1}.$ Now the
range of $F^{-1}D_{g}$ is a subset of
$\mathbb{W}^{1,p}_{-\epsilon-1-2/p}\subset
\mathbb{W}^{1,p}_{-\epsilon-1}.$ So the
range of $F^{-1}D_{g}$ is closed in
$\mathbb{W}^{1,p}_{-\epsilon-1-2/p}.$

This establishes that $F^{-1}D_{g}$ is semi-Fredholm. Now the
operator $P$ is a perturbation of $F^{-1}D_{g}$ in the operator
norm
\begin{equation*}
  ||P-F^{-1}D_{g}||_{\text{op}}=\sup\left\{||(P-F^{-1}D_{g})u||_{1,p,-\epsilon-1-2/p}:u\in
  \mathbb{W}^{2,p}_{-\epsilon},||u||_{2,p,-\epsilon}=1\right\}
\end{equation*}
because $||P-F^{1}D_{g}||_{\text{op},R}=o(1)$ as $R\rightarrow
\infty.$
Thus by the arguments of Theorem 1.10 of Bartnik \cite{Bart}
adapted to a first order operator (for general order see section 4 in Nirenberg
and Walker \cite{NW}), $P$ is semi-Fredholm. Since $D_{\hat{g}}$
is a perturbation of $P$ similar arguments shows that
\begin{equation} \label{iso}
D_{\hat{g}}:\mathbb{W}^{2,p}_{-\epsilon}\longrightarrow
\mathbb{W}^{1,p}_{-\epsilon-1-2/p}
\end{equation}
is semi-Fredholm onto the range of $D_{\hat{g}}.$ Next we show
that the kernel of the adjoint operator $D_{\hat{g}}^{\ast}$ is
trivial. Let $\psi\in \text{Ker }D_{\hat{g}}^{\ast}.$ After
extending $\psi$ to the whole of
$\mathbb{W}^{1,p}_{-\epsilon-1-2/p},$considering $\psi$ to be an
element of the dual space
$W^{2,\hat{p}}_{\epsilon+2/p-2},$ and using the formal
self-adjointness of the Dirac operator we see that $\psi\in
\text{Ker
}D_{\hat{g}}.$ (This dual space involves the norm
$||\cdot||_{2,\hat{p},\epsilon+2/p-2}$ relative to the
asymptotically flat metric. We recall that functions in
$W^{2,\hat{p}}_{\epsilon-2}$ with finite
$||\cdot||_{\hat{g},2,\hat{p},\epsilon-2}$ norm are in
$W^{2,\hat{p}}_{\epsilon+2/p-2}$ \ with finite
$||\cdot~||_{2,\hat{p},\epsilon+2/p-2}$ norm). But $\hat{g}$ has nonnegative scalar curvature and
 $\psi$ vanishes at infinity. So by virtue of the
Weitzenb\"ock-Lichnerowicz formula and maximum principle the
kernel of $D_{\hat{g}}^{\ast}$ is trivial and hence the
semi-Fredholm operator $D_{\hat{g}}$ is Fredholm. In fact
$\text{Ker }D_{\hat{g}}$ is also trivial. Thus
\eqref{iso}  is an isomorphism.
To generalize the result for more than one end we define the
Sobolev spaces with $r$ extended in $\Sigma$ so that it matches at
all ends with the asymptotic radial coordinates (for details see
p.230 in Parker and Taubes \cite{PT}). As stated before
$(\Sigma,\hat{g})$ is globally axisymmetric. However for this
proof, axisymmetry is necessary only at the ends of $\Sigma.$
\end{proof}

To get a suitable harmonic spinor we need the transformation
formulas stated in the following lemma.
\begin{lemma} \label{lemmaq}
If $\hat{g}=\zeta^{2}g$, then
$D_{\hat{g}}(\zeta^{-1}\xi)=\zeta^{-2}D_{g}\xi$. If the spinor
satisfies $\partial_{\phi}\xi=0$ and $\hat{g}=\bar{g}+fd\phi^{2},
g=\bar{g}+qfd\phi^{2}$, then
$D_{\hat{g}}(q^{-\frac{3}{8}}\xi)=q^{-\frac{3}{8}}D_{g}\xi$.
\end{lemma}
In the above lemma the first formula is the well-known conformal
transformation formula in $3$ dimension. A derivation of the
second formula can be found in \cite{Mas}.

Let $\xi_{0}$ be a spinor {\it constant near infinity } relative
to the asymptotically flat metric $g.$ In particular all partial
derivatives of $\xi_{0}$ vanishes in $U$ and we extend $\xi_{0}$
outside by keeping $\partial_{\phi}\xi_{0}=0.$ We define
$g_{1}=F^{2}g.$ Taking $q=F^{-4}$ in
\lemref{lemmaq} we get in $U$
\begin{equation} \label{qf4}
D_{\hat{g}}(q^{-\frac{3}{8}}F^{-1}\xi_{0})=q^{-\frac{3}{8}}D_{g_{1}}(F^{-1}\xi_{0})=q^{-\frac{3}{8}}F^{-2}D_{g}\xi_{0}.
\end{equation}
If we choose
$\Theta_{0}=q^{-\frac{3}{8}}F^{-1}\xi_{0}=F^{\frac{1}{2}}\xi_{0}$
then $D_{\hat{g}}\Theta_{0}\in \mathbb{W}^{1,p}_{-\epsilon-1}$ for
some $\epsilon \in (2/p,2-2/p).$ We get a unique spinor
$\widetilde{\Theta}\in
\mathbb{W}^{1,p}_{-\epsilon}$ satisfying
$D_{\hat{g}}\widetilde{\Theta}=D_{\hat{g}}\Theta_{0}.$ So we get a
nontrivial spinor $\Theta=\widetilde{\Theta}-\Theta_{0}$ harmonic
relative to $\hat{g}.$

\begin{remark} \label{rmMS1}
{\em In the case more than one end one usually chooses the
asymptotically constant spinor to be nonzero only at one end. We
shall then take $\Theta$ to be the average of the harmonic spinors
obtained.}
\end{remark}

\section{A positive mass type theorem}
\begin{theorem} \label{thm3}
Suppose $(\Sigma,\hat{g})$ is complete having a single end and the
scalar curvature $R_{\hat{g}}\geq 0.$ Let $\Theta$ be the harmonic
spinor for $\hat{g}$ as stated before \remref{rmMS1}. If $b>0,$
then
\begin{equation} \label{bfromR}
  b=\dfrac{3}{32\pi}\lim\limits_{r\rightarrow \infty}\left(r^{-3}\int\limits_{\hat{g},\mathcal{B}_{r}}
  \left[R_{\hat{g} }||\Theta||^{2}+4||\nabla _{\hat{g}
  }\Theta||^{2}\right]\right).
\end{equation}
\end{theorem}
\begin{proof} Let $n_{\hat{g}}$ be the unit normal form relative to
$\hat{g}$ on $\partial\mathcal{B}_{r}$ where $\mathcal{B}_{r}$ is
a ball of large radius $r$ in asymptotic region pointing in the
direction of increasing $r.$ By integrating the
Weitzenb\"ok-Lichnerowicz identity namely
\begin{equation}\label{lich}
  2\Delta_{\hat{g}} ||\Theta||^{2}=R_{\hat{g} }||\Theta||^{2}+4||\nabla _{\hat{g} }\Theta||^{2}
\end{equation}
we get
\begin{equation} \label{mass1}
\int\limits_{\hat{g},\partial\mathcal{B}_{r}}\dfrac{\partial||\Theta||^{2}}{\partial
n_{\hat{g}}}\geq 0.
\end{equation}
As before for some $\epsilon\in (2/p,2-2/p),$
$\widetilde{\Theta}=\Theta+F^{1/2}\xi_{0}\in W^{1,p}_{-\epsilon}$
where $\xi_{0}$ was defined before Eq.~\eqref{qf4}.
\begin{eqnarray*}
||\Theta||^{2}&=&||\widetilde{\Theta}||^{2}-\left\langle
\widetilde{\Theta},\Theta_{0}\right\rangle-\left\langle
\Theta_{0},\widetilde{\Theta}\right\rangle+||\Theta_{0}||^{2}\\
&=&||\widetilde{\Theta}||^{2}-F^{1/2}\left\langle
 \widetilde{\Theta},\xi_{0}\right\rangle- F^{1/2}\left\langle
\xi_{0},\widetilde{\Theta}\right\rangle+F
\end{eqnarray*}
since for the asymptotically constant spinor $||\xi_{0}||=1.$ Now
$\widetilde{\Theta}\in W^{2,p}_{-\epsilon}$, and
$n_{\hat{g}}=n_{\hat{g},r}dr=\sqrt{\overline{g}_{rr}}dr.$ So
\begin{eqnarray*}
  \dfrac{\partial||\Theta||^{2}}{\partial
n_{\hat{g}}} &=&\hat{g}^{jk}(||\Theta||^{2})_{j}n_{\hat{g},k}=
\sqrt{\overline{g}^{rr}}\dfrac{\partial||\Theta||^{2}}{\partial
r}=
\sqrt{\overline{g}^{rr}}\dfrac{\partial F}{\partial
r}+O(r^{-\epsilon})\\&=&2F^{-1}br\sin^{2}\theta+O(r^{-\epsilon}).
\end{eqnarray*}
Thus
\begin{eqnarray*}
\int\limits_{\hat{g},\partial\mathcal{B}_{r}}\dfrac{\partial||\Theta||^{2}}{\partial
n_{\hat{g}}}&=&2\pi\int\limits_{0}^{\pi}
\left(2F^{-1}br\sin^{2}\theta+O(r^{-\epsilon})\right)F
r^{2}\sin \theta d\theta \\
&=&2\pi\int\limits_{0}^{\pi}\left(2br^{3}\sin^{3}\theta+O(r^{4-\epsilon})\right)d\theta
=\frac{16}{3}\pi br^{3}+O(r^{4-\epsilon}).
\end{eqnarray*}
Since we can choose $\epsilon>1,$ from Eq.~\eqref{mass1} we get
Eq.~\eqref{bfromR}.
\end{proof}

\begin{remark} \label{rmMS2}
{\em In this remark we explain the underlying principle how the
harmonic spinor found in this paper can be used in the proof of
the uniqueness magnetized Schwarzschild solution in
\cite{Mas} where its existence is assumed. Suppose $\zeta^{\pm}$ and $f^{\pm}$ are positive
functions having decay as follows
\begin{eqnarray}
    \zeta^{\pm} &=& (1/4)(1-M/r\pm\sqrt{1-2M/r})^{2}F^{-2} \label{zeta}\\
  f^{\pm} &=& (1/4)(1-M/r\pm\sqrt{1-2M/r})^{2}r^{2}\sin^{2}\theta \label{f}
\end{eqnarray}
If now $\hat{g}$ has decay as in Eq.~\eqref{AF2} then the metric
$\eta^{+} = \zeta^{+} \bar{g}+f^{+}d\phi^{2}$ is asymptotically
flat with mass zero and $\eta^{-} = \zeta^{-}
\bar{g}+f^{-}d\phi^{2}$ compactifies the infinity. If the scalar
curvature $R_{\hat{g}}\geq 0$ is such that we also have
$R_{\eta^{+}}\geq 0$ then positive mass theorem says $\eta^{+}$ is
flat. $\eta^{-}$ is necessary because for this uniqueness problem
$\hat{g}$ is not complete but has a smooth totally geodesic
surface. Now $(\Sigma,\hat{g})$ has two ends and a totally
geodesic surface across which it is symmetric. If $\Theta$ be the
 the average harmonic spinor mentioned in
\remref{rmMS1}, then the normal derivative of $||\Theta||^{2}$ on the
totally geodesic surface vanishes. Unfortunately from the field
equations we cannot show $R_{\eta^{\pm}}\geq 0$ directly. But from
$\Theta$ we get harmonic spinors relative to $\eta^{\pm}$ and
exploiting their Weitzenb\"ok-Lichnerowicz identities we can build
two divergence form identities involving the expressions for
$R_{\eta^{\pm}}.$ On integration these identities give the
uniqueness result by virtue of the properties of $\Theta$ on the
totally geodesic surface and at infinity.}
\end{remark}

\section{Conclusion}
We showed the existence of harmonic spinors in some cylindrical
3-manifolds. Existence of such spinors are relevant in studying
the uniqueness or rigidity theorems involving cylindrical geometry
which is natural for a magnetic universe. Although the story is a
simple one that the index of the elliptic operator is the same
under some reasonable class of perturbations, in practice one
needs to check huge amount of computation before one can apply it
in a specific situation. Hopefully the method will be useful in
extending the black hole uniqueness theorems for static and
stationary solutions in a magnetic universe. We also hope that the
basic ideas of \textquotedblleft compensating geometry" presented
in this paper will inspire similar investigations in higher
dimensions involving generalized Dirac-type operators and metrics
not fully conformally related.

\bibliography{}

\address{$^{\flat}$ Mathematical Sciences Center, Tsinghua University \\
Haidian District, Beijing 100084, PRC.\\
\email{abulm@math.tsinghua.edu.cn} \\
$^{\sharp}$School of Mathematical Sciences, Fudan University\\
Yangpu District, Shanghai 200433, PRC.\\
\email{qizhiwang12@fudan.edu.cn} }

\end{document}